\documentclass[runningheads,a4paper]{llncs}
\usepackage{tikz}
\usetikzlibrary{automata,positioning}
\usetikzlibrary{arrows,automata}

\usepackage{amsmath,latexsym}
\usepackage{amsfonts}
\usepackage{amssymb}
\usepackage{amssymb}
\usepackage{graphicx}
\usepackage{epstopdf}
\usepackage{subfigure}
\usepackage{setspace} 
\usepackage{authblk}
\usepackage{algpseudocode}
\usepackage{algorithm}

\def\qedbox#1#2{\vbox{\hrule height.2pt
  \hbox{\vrule width.2pt height#2pt \kern#1pt \vrule width.2pt}
  \hrule height.2pt}}
\def\qed{\hfill \quad\qedbox46\newline\smallbreak}

\def\s#1{\mbox{\boldmath $#1$}}

\def\ceil#1{\lceil #1 \rceil}

\def\+{\!+\!}
\def\-{\!-\!}

\def\itbf#1{\textit{\textbf{#1}}}
\def\match{\approx}
\def\cP{\mathcal{P}}

\def\bproc{{\bf procedure\ }}
\def\bfunc{{\bf function\ }}

\def\bfor{{\bf for\ }}
\def\bto{{\bf to\ }}

\def\bwhile{{\bf while\ }}

\def\band{{\bf and\ }}

\def\bdo{{\bf do\ }}
\def\bif{{\bf if\ }}
\def\bthen{{\bf then\ }}

\def\bnot{{\bf not\ }}

\def\la{\leftarrow}

\def\qq{\qquad}
\def\com#1{\hspace{24pt}{\bf $\triangleright$}\hspace{6pt}{\sl #1}}

\def\pref(#1,#2){$#1$ is a prefix of $#2$}
\def\suff(#1,#2){$#1$ is a suffix of $#2$}

\def\reg(#1,#2){$#2$ is $#1$-regular}
\def\notreg(#1,#2){$#2$ is not $#1$-regular}
\def\top{\tt{top}}
\def\pop{\tt{pop}}
\def\true{\tt{true}}
\def\false{\tt{false}}
\def\UPDATE\_F{\tt{UPDATE\_F}}
\def\LEAST{\tt{LEAST}}
\def\MERGE{\tt{MERGE}}

\newif\ifProofs
\Proofsfalse

\newif\ifRev
\Revfalse

\begin{document}

\pagestyle{headings}

\title {Inferring an Indeterminate String \\ from a Prefix Graph}
\titlerunning{\itshape{Inferring an Indeterminate String from a Prefix
Graph}}
 
\author{
Ali Alatabbi\inst{1}
\and
M.\ Sohel Rahman\thanks{Supported by an ACU Titular Fellowship.}\inst{2}
\and
W.\ F.\ Smyth\thanks{Supported in part by the Natural Sciences \& Engineering
Research Council of Canada.}\inst{3,4}}

\authorrunning{\itshape{Alatabbi, Rahman and Smyth.}}

\institute{Department of Informatics, King's College London\\
\email{ali.alatabbi@kcl.ac.uk}
\and Department of Computer Science \& Engineering\\
Bangladesh University of Engineering \& Science\\
\email{msrahman@cse.buet.ac.bd}
\and Algorithms Research Group, Department of Computing \& Software\\
McMaster University\\
\and School of Engineering \& Information Technology, \\
Murdoch University\\
\email{smyth@mcmaster.ca}}

\maketitle

\begin{abstract}
An \itbf{indeterminate string} (or, more simply, just a \itbf{string})
$\s{x} = \s{x}[1..n]$ on an alphabet $\Sigma$
is a sequence of nonempty subsets of $\Sigma$.
We say that $\s{x}[i_1]$ and $\s{x}[i_2]$ \itbf{match}
(written $\s{x}[i_1] \match \s{x}[i_2]$) if and only if
$\s{x}[i_1] \cap \s{x}[i_2] \ne \emptyset$.
A \itbf{feasible array} is an array $\s{y} = \s{y}[1..n]$ of integers
such that $\s{y}[1] = n$ and for every $i \in 2..n$,
$\s{y}[i] \in 0..n\- i\+ 1$.
A \itbf{prefix table} of a string $\s{x}$ is an array $\s{\pi} = \s{\pi}[1..n]$
of integers such that, for every $i \in 1..n$,
$\s{\pi}[i] = j$ if and only if $\s{x}[i..i\+ j\- 1]$
is the longest substring at position $i$ of \s{x} that matches a prefix of \s{x}.
It is known from \cite{CRSW13} that every feasible array is a prefix table of
some indetermintate string.
A \itbf{prefix graph} $\mathcal{P} = \mathcal{P}_{\s{y}}$
is a labelled simple graph whose structure is determined
by a feasible array \s{y}.
In this paper we show, given a feasible array \s{y},
how to use $\mathcal{P}_{\s{y}}$ to construct
a lexicographically least indeterminate string
on a minimum alphabet whose prefix table $\s{\pi} = \s{y}$.
\end{abstract}

\section{Introduction}
\label{sect-intro}
In the extensive literature of stringology/combinatorics on words,
a ``string'' or ``word'' has usually been defined as a sequence
of individual elements of a distinguished set $\Sigma$ called an ``alphabet''.
Nevertheless, going back as far as the groundbreaking paper of
Fischer \& Paterson \cite{FP74}, more general sequences,
defined instead on {\it subsets} of $\Sigma$, have also been considered.
The more constrained model introduced in \cite{FP74} restricts
entries in a string to be either elements of $\Sigma$
(subsets of size 1) or $\Sigma$ itself (subsets of size $\sigma = |\Sigma|$);
these have been studied in recent years as
``strings with don't cares'' \cite{IMMP03}, also
``strings with holes'' or ``partial words'' \cite{BS08}.
The unconstrained model, which allows arbitrary nonempty subsets of $\Sigma$,
has also attracted significant attention,
often because of applications in bioinformatics:
such strings have variously been called ``generalized'' \cite{A87},
``indeterminate'' \cite{HS03}, or ``degenerate'' \cite{IMR08}.

In this paper we study strings
in their full generality,
hence the following definitions:

\begin{definition}
\label{defn-string}
Suppose a set $\Sigma$ of symbols (called the \itbf{alphabet}) is given.
A \itbf{string} \s{x} on $\Sigma$ of \itbf{length} $n = |\s{x}|$ is a sequence
of $n \ge 0$ nonempty finite subsets of $\Sigma$, called \itbf{letters};
we represent \s{x} as an array $\s{x}[1..n]$.
If $n = 0$, \s{x} is called the \itbf{empty string}
and denoted by $\s{\varepsilon}$;
if for every $i \in 1..n$, $\s{x}[i]$ is a subset of $\Sigma$ of size 1,
\s{x} is said to be a \itbf{regular string}.
\end{definition}

\begin{definition}
\label{defn-match}
Suppose we are given two strings \s{x} and \s{y} and integers $i \in 1..|\s{x}|$,
$j \in 1..|\s{y}|$. We say that $\s{x}[i]$ and $\s{y}[j]$ \itbf{match}
(written $\s{x}[i] \match \s{y}[j]$) if and only if $\s{x}[i] \cap \s{y}[j] \ne \emptyset$. Then \s{x} and \s{y} \itbf{match} ($\s{x} \match \s{y}$)
if and only if $|\s{x}| = |\s{y}|$
and $\s{x}[i] \match \s{y}[i]$ for every $i \in 1..|\s{x}|$.
\end{definition}
Note that matching is not necessarily transitive:
$a \match \{a,b\} \match b$, but $a \not\match b$.

\begin{definition}
\label{defn-prefix}
The \itbf{prefix table} (also \itbf{prefix array})\footnote{
We prefer ``table'' because of the possible confusion with ``suffix array'',
a completely different data structure.}
of a string $\s{x} = \s{x}[1..n]$ is the integer array
$\s{\pi}_{\s{x}} = \s{\pi}_{\s{x}}[1..n]$
such that for every $i \in 1..n$, $\s{\pi}_{\s{x}}[i]$ is the
length of the longest prefix of $\s{x}[i..n]$ that matches a prefix of $\s{x}$.
Thus for every prefix table $\s{\pi}_{\s{x}}$, $\s{\pi}_{\s{x}}[1] = n$.
When there is no ambiguity, we write $\s{\pi} = \s{\pi}_{\s{x}}$.
\end{definition}
The prefix table is an important data structure for strings:
it identifies all the borders, hence all the periods,
of every prefix of \s{x} \cite{CRSW13}.
It was originally introduced to facilitate the computation of
repetitions in regular strings \cite{ML84}, see also \cite{S03};
and for regular strings, prefix table and border array are equivalent,
since each can be computed from the other in linear time \cite{BKS13}.
For general strings, the prefix table can be computed in compressed form in $O(n^2)$ time
using $\Theta(n/\sigma)$ bytes of storage space \cite{SW08},
where $\sigma = |\Sigma|$.
Two examples follow, adapted from \cite{CRSW13}:
\begin{equation}
\label{pi1}
\begin{array}{rcccccccc}
& \scriptstyle 1 & \scriptstyle 2 & \scriptstyle 3 & \scriptstyle 4 & \scriptstyle 5 & \scriptstyle  6 & \scriptstyle 7 & \scriptstyle 8 \\
\s{x_1} = & a & c & a & g & a & c & a & t \\
\s{\pi_1} = & 8 & 0 & 1 & 0 & 3 & 0 & 1 & 0
\end{array}
\end{equation}
\begin{equation}
\label{pi2}
\begin{array}{rcccccccc}
& \scriptstyle 1 & \scriptstyle 2 & \scriptstyle 3 & \scriptstyle 4 & \scriptstyle 5 & \scriptstyle  6 & \scriptstyle 7 & \scriptstyle 8 \\
\s{x_2} = & \{a,c\} & \{g,t\} & \{a,g\} & \{a,c,g\} & g & c & \{a,t\} & a \\
\s{\pi_2} = & 8 & 0 & 4 & 2 & 0 & 3 & 1 & 1
\end{array}
\end{equation}

Since clearly every position $i \in 2..n$ in a prefix table \s{\pi}
must satisfy $0 \le \s{\pi}[i] \le n\- i\+ 1$, the following definition is a natural one:
\begin{quote}
An array $\s{y} = \s{y}[1..n]$ of integers is said to be a \itbf{feasible array}
if and only if $\s{y}[1] = n$ and for every $i \in 2..n$,
$\s{y}[i] \in 0..n\- i\+ 1$.
\end{quote}

An immediate consequence of Definition~\ref{defn-prefix} is the following:
\begin{lemma}[\cite{CRSW13}]
\label{lemm-easy}
Let $\s{x} = \s{x}[1..n]$ be a string.
An integer array $\s{y} = \s{y}[1..n]$ is the prefix table
of $\s{x}$ if and only if for each position $i \in 1..n$,
the following two conditions hold:
\begin{itemize}
\item[(a)] $\s{x}\big[1..\s{y}[i]\big] \match \s{x}\big[i..i+ \s{y}[i]-1\big]$;
\item[(b)] if $i + \s{y}[i] \le n$, then $\s{x}\big[y[i]+1\big] \not\match \s{x}\big[i +\s{y}[i]\big]$.
\end{itemize}
\end{lemma}
Then the following fundamental result establishes the important connection
between strings and feasible arrays:

\begin{lemma}[\cite{CRSW13}]
\label{lemm-feas}
Every feasible array is the prefix table of some string.
\end{lemma}

In view of this lemma, we say that a feasible array
is \itbf{regular} if it is the prefix array of a regular string.
We are now able to state the goal of this paper as follows:
for a given feasible array $\s{y} = \s{y}[1..n]$,
not necessarily regular,
construct a string \s{x} on a minimum alphabet
whose prefix table $\s{\pi}_{\s{x}} = \s{y}$ --- the ``reverse engineering''
problem for the prefix table in its full generality.
In fact, we do somewhat more: we construct a
lexicographically least such string,
in a sense to be defined in Section~\ref{sect-prelim}.

The first reverse engineering problem in stringology was stated and solved
in \cite{FLRS99,FGLR02}, where an algorithm was described
to compute a lexicographically least regular string whose border array was
a given integer array --- or to return the result that no such regular string exists.
Many other similar constructions have since been published,
related to other stringological data structures but always specific to
regular strings (see \cite{BIST03,DLL05,FS06,MNRR13}, among others).
\cite{NRR12} was the first paper to consider the more general problem
of inferring an indeterminate string from a given data structure
(specifically, border array, suffix array and LCP array).
Although solving such problems does not yield immediate applications,
nevertheless solutions provide a deeper understanding of the combinatorial many-many relationship between
strings and the various data structures developed from them
(see for example \cite{MSM99},
where canonical strings corresponding to given border arrays are identified
and efficiently generated for use as test data).

For prefix tables and regular strings,
the reverse engineering problem was solved in \cite{CCR09},
where a linear-time algorithm was described to return
a lexicographically least regular string \s{x} whose
prefix table is the given feasible array \s{y}, or an error message if no
such \s{x} exists.
A recent paper \cite{BBD14} sketches two algorithms
to compute an indeterminate string \s{x} on a minimum alphabet
(not necessarily lexicographically least) corresponding to
a given feasible array \s{y},
but the algorithms are theoretical in nature:
one requires the determination of the chromatic number of a certain graph,
an NP-hard problem,
while the other depends on somehow identifying the minimum
``induced positive edge cover'' of a graph.
However, \cite{BBD14} proves an important result that we use
below to bound the complexity of our algorithm:
that the minimum alphabet size $\sigma$ of a string corresponding
to a given feasible array of length $n$ is at most $n\+ \sqrt{n}$.
In this paper
we use graph-theoretic methods developed from \cite{CRSW13}
to compute a lexicographically least string,
regular or not, corresponding to the given \s{y},
in time $O(\sigma n^2)$.


Section~\ref{sect-prelim} of this paper provides background material
for an understanding of our algorithm;
Section~\ref{sect-alg} presents the algorithm itself;
Section~\ref{sect-future} briefly discusses these results
and suggests future work.

\section{Preliminaries}
\label{sect-prelim}
Following \cite{CRSW13}, for a given feasible array $\s{y} = \s{y}[1..n]$,
we define a corresponding graph $\cP = \cP_{\s{y}}$,
on which our algorithm will be based:

\ifRev
\begin{figure}[htbp]
\label{fig-assign}
  \begin{minipage}{0.5\linewidth}
  \centering
  \includegraphics*[viewport = 1.5cm 12cm 9cm 19cm,scale=0.7]{ex-1-p}
  \caption{$\mathcal{P}_{\s{y_1}}^+$ for $\s{y_1} = 80103010$}\label{Graph-ex-1-p}
  \end{minipage}
  \hfill
  \begin{minipage}{0.5\linewidth}
  \centering
  \includegraphics*[viewport = 1.5cm 12cm 9cm 19cm,scale=0.7]{ex-1-n}\\
  \caption{$\mathcal{P}_{\s{y_1}}^-$ for $\s{y_1} = 80103010$}\label{Graph-ex-1-n}
  \end{minipage}
\end{figure}
\fi

\begin{definition}
\label{defn-P}
Let $\mathcal{P} = (V,E)$ be a labelled graph with vertex set $V =
\{1,2,\ldots,n\}$ consisting of positions in a given feasible array
$\s{y}.$ In $\mathcal{P}$ we define, for $i \in 2..n$, two kinds of edge
(compare Lemma~\ref{lemm-easy}):
\begin{itemize}
\item[(a)]
for every $h \in 1..\s{y}[i]$, $(h,i\+ h\- 1)$ is called a \itbf{positive edge};
\item[(b)]
$(1\+ \s{y}[i],i\+ \s{y}[i])$ is called a \itbf{negative edge}, provided $i\+ \s{y}[i] \le n.$
\end{itemize}
$E^+$ and $ E^-$ denote the sets of positive and negative edges, respectively.  We write
$E = E^+ \cup E^-$, $\mathcal{P}^+ = (V,E^+)$,
$\mathcal{P}^- = (V,E^-)$,
and we call $\mathcal{P}$ the \itbf{prefix graph}
$\cP_{\s{y}}$ of \s{y}.
If $\s{x}$ is a string having $\s{y}$
as its prefix table,
then we also refer to $\mathcal{P}$ as the \itbf{prefix graph}
$\cP_{\s{x}}$ of \s{x}.
\end{definition}
Observe that $E^+$ and $E^-$ are necessarily disjoint.
\ifRev
Figures \ref{Graph-ex-1-p}--\ref{Graph-ex-2-n} show the prefix graphs,
as given in \cite{CRSW13}, for the example strings (\ref{pi1}) and (\ref{pi2}).
Again, in
\fi
Figures \ref{Graph-ex-1-p-New}--\ref{Graph-ex-2-n-New} show the prefix graphs
\ifRev of two different indeterminate
\fi
for the example strings (\ref{piN1}) and (\ref{piN2}).

\begin{equation}
\label{piN1}
\begin{array}{rcccccccc}
& \scriptstyle 1 & \scriptstyle 2 & \scriptstyle 3 & \scriptstyle 4 & \scriptstyle 5 & \scriptstyle  6 & \scriptstyle 7 & \scriptstyle 8 \\
\s{x_3} = & \{a,b\} & \{a,c\} & c & \{a,b\} & b & c & \{a,c\} & b \\
\s{\pi_3} = & 8 & 2 & 0 & 1 & 4 & 0 & 1 & 1
\end{array}
\end{equation}

\begin{equation}
\label{piN2}
\begin{array}{rcccccccc}
& \scriptstyle 1 & \scriptstyle 2 & \scriptstyle 3 & \scriptstyle 4 & \scriptstyle 5 & \scriptstyle  6 & \scriptstyle 7 & \scriptstyle 8 \\
\s{x_4} = & \{a,b\} & \{a,c\} & \{a,d\} & \{c,e\} & a & \{b,e\} & c & d \\
\s{\pi_4} = & 8 & 2 & 4 & 0 & 1 & 3 & 0 & 0
\end{array}
\end{equation}

\ifRev
\begin{figure}[htbp]
  \begin{minipage}{0.5\linewidth}
  \centering
  \includegraphics*[viewport = 1.5cm 12cm 9cm 19cm,scale=0.7]{ex-2-p}\\
  \caption{$\mathcal{P}_{\s{y_2}}^+$ for $\s{y_2} = 80420311$}\label{Graph-ex-2-p}
  \end{minipage}
  \hfill
  \begin{minipage}{0.5\linewidth}
  \centering
  \includegraphics*[viewport = 1.5cm 12cm 9cm 19cm,scale=0.7]{ex-2-n}\\
  \caption{$\mathcal{P}_{\s{y_2}}^-$ for $\s{y_2} = 80420311$}\label{Graph-ex-2-n}
  \end{minipage}
\end{figure}
\fi

\begin{figure}[htbp]
\label{fig-assign-ex1}
  \begin{minipage}{0.5\linewidth}
  \centering
\begin{tikzpicture}	[scale=.7]
    \path (0,0)   node[circle,draw](1) {1};
    \path (-2,-1) node[circle,draw](2) {2};
    \path (-3,-3) node[circle,draw](3) {3};
    \path (-2,-5) node[circle,draw](4) {4};
    \path (0,-6) node[circle,draw](5) {5};
    \path (2,-5) node[circle,draw](6) {6};
    \path (3,-3) node[circle,draw](7) {7};
    \path (2,-1) node[circle,draw](8) {8};
    \draw[-] (1) -- (2);
    \draw[-] (1) -- (4);
    \draw[-] (1) -- (5);
    \draw[-] (1) -- (7);
    \draw[-] (1) -- (8);
    \draw[-] (2) -- (3);
    \draw[-] (2) -- (6);
    \draw[-] (3) -- (7);
    \draw[-] (4) -- (8);
\end{tikzpicture}
  \caption{$\mathcal{P}_{\s{y_3}}^+$ for $\s{y_3} = 82014011$}\label{Graph-ex-1-p-New}
  \end{minipage}
  \hfill
  \begin{minipage}{0.5\linewidth}
  \centering
  \begin{tikzpicture}	[scale=.7]
    \path (0,0)   node[circle,draw](1) {1};
    \path (-2,-1) node[circle,draw](2) {2};
    \path (-3,-3) node[circle,draw](3) {3};
    \path (-2,-5) node[circle,draw](4) {4};
    \path (0,-6) node[circle,draw](5) {5};
    \path (2,-5) node[circle,draw](6) {6};
    \path (3,-3) node[circle,draw](7) {7};
    \path (2,-1) node[circle,draw](8) {8};
    \draw[-] (1) -- (3);
    \draw[-] (1) -- (3);
    \draw[-] (2) -- (5);
    \draw[-] (2) -- (8);
    \draw[-] (3) -- (4);
    \draw[-] (6) -- (1);
    \draw[-] (8) -- (2);

\end{tikzpicture}
  \caption{$\mathcal{P}_{\s{y_3}}^-$ for $\s{y_3} = 82014011 $}\label{Graph-ex-1-n-New}
  \end{minipage}
\end{figure}

\begin{figure}[htbp]
\label{fig-assign-ex2}
  \begin{minipage}{0.5\linewidth}
  \centering
\begin{tikzpicture}	[scale=.7]
    \path (0,0)   node[circle,draw](1) {1};
    \path (-2,-1) node[circle,draw](2) {2};
    \path (-3,-3) node[circle,draw](3) {3};
    \path (-2,-5) node[circle,draw](4) {4};
    \path (0,-6) node[circle,draw](5) {5};
    \path (2,-5) node[circle,draw](6) {6};
    \path (3,-3) node[circle,draw](7) {7};
    \path (2,-1) node[circle,draw](8) {8};
    \draw[-] (1) -- (2);
    \draw[-] (1) -- (3);
    \draw[-] (1) -- (5);
    \draw[-] (1) -- (6);
    \draw[-] (2) -- (3);
    \draw[-] (2) -- (4);
    \draw[-] (2) -- (7);
    \draw[-] (3) -- (8);
    \draw[-] (3) -- (5);
    \draw[-] (4) -- (6);
    \draw[-] (5) -- (1);
\end{tikzpicture}
  \caption{$\mathcal{P}_{\s{y_4}}^+$ for $\s{y_4} = 82401300$}\label{Graph-ex-2-p-New}
  \end{minipage}
  \hfill
  \begin{minipage}{0.5\linewidth}
  \centering
  \begin{tikzpicture}	[scale=.7]
    \path (0,0)   node[circle,draw](1) {1};
    \path (-2,-1) node[circle,draw](2) {2};
    \path (-3,-3) node[circle,draw](3) {3};
    \path (-2,-5) node[circle,draw](4) {4};
    \path (0,-6) node[circle,draw](5) {5};
    \path (2,-5) node[circle,draw](6) {6};
    \path (3,-3) node[circle,draw](7) {7};
    \path (2,-1) node[circle,draw](8) {8};
    \draw[-] (1) -- (4);
    \draw[-] (1) -- (7);
    \draw[-] (1) -- (8);
    \draw[-] (2) -- (6);
    \draw[-] (3) -- (4);
    \draw[-] (5) -- (7);
\end{tikzpicture}
  \caption{$\mathcal{P}_{\s{y_4}}^-$ for $\s{y_4} = 82401300 $}\label{Graph-ex-2-n-New}
  \end{minipage}
\end{figure}

The following lemma will be useful for the analysis of our algorithm:

\begin{lemma}[\cite{CRSW13}]
\label{lemm-reg}
Let $\mathcal P_{\s{y}} = (V, E)$ be a prefix graph of a feasible array
\s{y}. Then \s{y} is regular if and only if every edge of $P^-_{\s{y}}$
joins two vertices in distinct connected components of $P^+_{\s{y}}$.
\end{lemma}


So as to discuss the lexicographical ordering of strings on an ordered alphabet $\Sigma$,
we need first of all a definition of the order of two letters:
\begin{definition}
\label{letterorder}
Suppose two letters $\lambda$ and $\mu$ are given, where
$$\lambda = \{\lambda_1, \lambda_2, \ldots, \lambda_j\},\
\mu = \{\mu_1, \mu_2, \ldots, \mu_k\},$$
with $\lambda_h \in \Sigma$ for every $h \in 1..j$,
$\mu_h \in \Sigma$ for every $h\in 1..k$.
We assume without loss of generality that $j \le k$,
also that $\lambda_h < \lambda_{h+1}$ for every $h \in 1..j\- 1$ and
$\mu_h < \mu_{h+1}$ for every $h \in 1..k\- 1$.
Then $\lambda = \mu$ if and only if
$\lambda_h = \mu_h$ for every $h \in 1..k$ and $j=k$; while
$\lambda \prec \mu$ if and only if
\begin{itemize}
\item[(a)]
$\lambda_h = \mu_h$ for every $h \in 1..j < k$; or
\item[(b)]
$\lambda_h = \mu_h$ for every $h \in 1..h' < j$ {\bf and}
$\lambda_{h'+1} < \mu_{h'+1}$.
\end{itemize}
Otherwise, $\mu \prec \lambda$.
\end{definition}
Note that $(\lambda = \mu) \Rightarrow (\lambda \match \mu)$,
but that $\lambda \match \mu$ implies neither equality nor an ordering
of $\lambda$ and $\mu$.
We remark also that the definition of letter order given here is
not the only possible or useful one.
For example, it would require
$\{a,b,w,x,y,z\} \prec \{a,c\}$, thus arguably not
placing sufficient emphasis on economy of letter selection in the alphabet.

\begin{definition}
\label{wordorder}
Now suppose that two strings $\s{x_1} = \s{x_1}[1..n_1]$
and $\s{x_2} = \s{x_2}[1..n_2]$ on $\Sigma$ are given,
where without loss of generality we assume that $n_1 \le n_2$.
Then $\s{x_1} = \s{x_2}$ if and only if $\s{x_1}[h] = \s{x_2}[h]$
for every $h \in 1..n_2$ and $n_1=n_2$; while
$\s{x_1} \prec \s{x_2}$ if and only if
\begin{itemize}
\item[(a)]
$\s{x_1}[h] = \s{x_2}[h]$ for every $h \in 1..n_1 < n_2$; or
\item[(b)]
$\s{x_1}[h] = \s{x_2}[h]$ for every $h \in 1..h' < n_1$ \bf{and}
$\s{x_1}[h'\+ 1] \prec \s{x_2}[h'\+ 1]$.
\end{itemize}
Otherwise, $\s{x_2} \prec \s{x_1}$.
\end{definition}

To better illustrate the relation of strings defined and used in this paper
we present the following examples:
\begin{itemize}
  \item $\s{x_1}=\{a,c\}~\{g,t\}~a  \prec \s{x_2} =\{a,c\}~\{g,t\}~\{a,g\}$
  \item $\s{x_1}=a~\{g,t\}~\{a,c\}~\{a,c,g\}  \prec \s{x_2} =a~\{g,t\}~\{a,t\}~a$
  \item $\s{x_1}=a~\{a,c,g\}~\{a,c,g\}~\{a,t\} \prec \s{x_2} =~\{a,c\}~g~g~\{a,t\}$
\end{itemize}
where $a < c < g < t$.
\section{Algorithm RevEng}
\label{sect-alg}

\subsection{The Algorithm}
The basic strategy of Algorithm RevEng, that constructs a lexicographically least
string \s{x} (initially empty)
corresponding to a given feasible array $\s{y} = \s{y}[1..n]$, is
expressed by the main steps given below.
Initially the alphabet $\Sigma$ is empty ($\sigma = 0$), as are the sets $\s{x}[i]$,
$1 \le i \le n$.
\begin{description}
\item[(S1)]
Consider the edges $(i,j)$ of $E^+$ in increasing order of $ni\+ j$
in order to add a single letter to $\s{x}[i]$, $\s{x}[j]$, or both based on the following steps;
\item[(S2)]
if, by virtue of previous assignments, $\s{x}[i] \match \s{x}[j]$
(so neither is empty),
there is nothing to do --- $(i,j)$ can be skipped;
\item[(S3)]
otherwise, for the current $(i,j)$, determine a sequence
$$C = (\lambda_1,i_1),(\lambda_2,i_2),\ldots,(\lambda_r,i_r)$$
of all candidate assignments,
where for every $h \in 1..r$, $i_h = i$ (respectively, $j$)
if $\lambda_h \in \s{x}[j]$ (respectively, $\s{x}[i]$),
and $\lambda_1 < \lambda_2 < \cdots < \lambda_r$;
\item[(S4)]
for the current $h$, determine whether or not the assignment
$$\s{x}[i_h] \la \s{x}[i_h] \cup \{\lambda_h\}$$
is ``allowable'' (that is, compatible with
the neighbourhood of $i_h$ in $E^-$) --- if so,
then perform the assignment,
maintaining the elements of $\s{x}[i_h]$ in their natural order;
\item[(S5)]
if for no $h$ is the assignment allowable,
then assign a least new letter (drawn WLOG from the set of positive integers)
to both $\s{x}[i]$ and $\s{x}[j]$;
\item[(S6)]
since it may be that after Steps (S1)-(S5) have been executed
for every $(i,j) \in E^+$, there still remain unassigned positions in \s{x}
(that is, corresponding to isolated vertices in $\mathcal{P}^+$),
a final assignment of a least possible letter for these positions
is required (see function $\LEAST(i,\lambda_{\max})$ and Lemma~\ref{lemm-unassigned}).
\end{description}

In order to implement this algorithm, several data structures need
to be created, maintained, and accessed:

\begin{description}
\item[(DS1)]
The edges of $E^+$ are made accessible in increasing order for Step (S1) by
a radix sort of the positive edges $(i,j)$ into a linked list $L^+$ (i.e., the linked list $L^+$ contains the edges of $E^+$ in increasing order),
whose entries occur in increasing order of $i$ and, within each $i$,
in increasing order of $j$.
The time requirement is $\Theta(|E^+|)$, thus $O(n^2)$ in the worst case,
since $E^+$ can contain $\Theta(n^2)$ edges \cite{CRSW13}.
\item[(DS2)]
In order to implement Step (S4) of Algorithm RevEng, we need,
for each position $i \in 1..n$, to have available a linked list of positions
$j$ such that $(i,j)$ is an edge of $E^-$.
This can be done by using $E^-$ to form a set of negative edges that includes
each $(i,j)$ twice, both as $(i,j)$ and as $(j,i)$.
Then in a preprocessing step the entries in this set are radix sorted
into an array $N^- = N^-[1..n]$ of $n$ linked lists, such that for every $i \in 1..n$,
$N^-[i]$ gives in increasing order all the vertices $j$
for which $(i,j) \in E^-$
(in other words, the neighbourhood of $i$ in $E^-$).
Since $E^-$ contains $O(n)$ edges \cite{CRSW13},
this preprocessing step can be accomplished in $O(n)$ time.
\item[(DS3)]
Steps (S2)-(S5) require that for each $i \in 1..n$,
a linked list $\s{x}[i]$ be maintained of letters $\lambda$
that have so far been assigned to $\s{x}[i]$.
Each list is maintained in increasing letter order,
so that update, intersection, and union each require $O(\sigma)$ time,
where $\sigma = |\Sigma|$ is the (current) alphabet size.
Since for regular strings each $\s{x}[i]$ has exactly one element,
in this case processing time reduces to $O(1)$.
\item[(DS4)]
In Step (S4), in order to determine whether a proposed assignment
of a letter $\lambda_h$ to a position $i'_h$ in \s{x} is allowable or not,
we form a ``forbidden'' matrix $F[1..n,1..\sigma]$ in which
$F[i,\lambda] = 1$ if and only if $\lambda \in \s{x}[i]$ is forbidden.
$F$ is updated and used as follows:
\begin{itemize}
\item
for each new letter $\lambda_{\max}$ introduced in Step (S5),
$F[i,\lambda_{\max}]$ is initialized to zero for all $i \in 1..n$;
\item
whenever an assignment $\s{x}[i] \stackrel{+}{\la} \lambda$
is made in Steps (S4) \& (S5), set
$F[j,\lambda] \la 1$ for every $j \in N^-[i]$
(procedure \UPDATE\_F$(i,\lambda)$).
\end{itemize}
\end{description}
Figures~\ref{alg-assign} and \ref{alg-least} give pseudocode
for Algorithm RevEng and function LEAST, respectively.

\begin{figure}[h]
{\leftskip=0.35in\obeylines\sfcode`;=3000
\bproc RevEng $(\cP,\s{x},n)$
$\lambda_{\max} \la 0;\ \s{x} \la \emptyset^n;\ F[1..n, 1..\sigma] \la 0^{n \sigma}$
\bwhile \top$(L^+) \ne \emptyset$ \bdo
\qq $(i,j) \la$ \pop$(L^+);\ C \la \emptyset$ \com{$i < j;\ ni\+ j$ a minimum}
\qq \bif $\s{x}[i] \cap \s{x}[j] = \emptyset$ \bthen
\qq\qq $\forall \lambda \in \s{x}[i]$ \bdo $ C_1 \stackrel{+}{\la} (\lambda,j)$ \com{ordered by $\lambda$}
\qq\qq $\forall \lambda \in \s{x}[j]$ \bdo $ C_2 \stackrel{+}{\la} (\lambda,i)$ \com{ordered by $\lambda$}
\com{Merge $C_1$ and $C_2$ into a single sequence ordered by $\lambda$.}
\qq\qq $C \la \MERGE(C_1,C_2)$
\qq\qq $SET \la \false$
\qq\qq \bwhile \top$(C) \ne \emptyset$ \band \bnot $SET$ \bdo
\qq\qq\qq $(\lambda,h) \la$ \pop$(C)$
\qq\qq\qq \bif $F[h,\lambda] \ne 1$ \bthen
\qq\qq\qq\qq $\s{x}[h] \stackrel{+}{\la} \lambda$ \com{maintain $\lambda$ ordering}
\qq\qq\qq\qq $SET \la$ \true;\ \UPDATE\_F$(h,\lambda)$
\qq\qq \bif \bnot $SET$ \bthen
\qq\qq\qq $\lambda_{\max} \la \lambda_{\max}\+ 1$
\qq\qq\qq \bfor $h \la 1$ \bto $n$ \bdo $F[h,\lambda_{\max}] \la 0$
\qq\qq\qq $\s{x}[i] \stackrel{+}{\la} \lambda_{\max};\ $\UPDATE\_F$(i,\lambda_{\max})$
\qq\qq\qq $\s{x}[j] \stackrel{+}{\la} \lambda_{\max};\ $\UPDATE\_F$(j,\lambda_{\max})$
\bfor $i \la 1$ \bto $n$ \bdo
\qq \bif $\s{x}[i] = \emptyset$ \bthen
\com{Identify the least letter $\lambda$ that does \bnot occur}
\com{in {\bf any} $\s{x}[j]$ for which $j \in N^-[i]$.}
\qq\qq $\lambda \la \LEAST(i,\lambda_{\max});\ \lambda_{\max} \la \max(\lambda,\lambda_{\max})$
\qq\qq $\s{x}[i] \la \lambda$
}
\caption{Given the preprocessing outlined in (DS1)-(DS2), Algorithm RevEng computes $\s{x}[1..n]$, the lexicographically least string corresponding to a given prefix (feasible) graph $\cP$ on $n$ vertices.}
\label{alg-assign}
\end{figure}

\begin{figure}[htpb]
{\leftskip=1.5in\obeylines\sfcode`;=3000
\bfunc $\LEAST(i,\lambda_{\max})$
$B[1..\lambda_{\max}] \la 0^{\lambda_{\max}}$
$\forall j \in N^-[i]$ \bdo
\qq $\forall \lambda \in \s{x}[j]$ \bdo
\qq\qq $B[\lambda] \la 1$
$\lambda \la 1$
\bwhile $\lambda \le \lambda_{\max}$ \band $B[\lambda] = 1$ \bdo
\qq $\lambda \la \lambda\+ 1$
}
\caption{Identify the least letter $\lambda$ that does \bnot occur in {\bf any} $\s{x}[j]$ for which $j \in N^-[i]$.}
\label{alg-least}
\end{figure}

\subsection{Correctness}
Consider first the main \bwhile loop of Algorithm ASSIGN,
in which the edges of $E^+$ are considered in strict increasing $(i,j)$ order.
We see that new letters $\lambda_{\max}$ are first introduced
at the leftmost possible positions in $\s{x}$.
Thereafter, whenever a letter is reused ($\lambda_{\max}$ not increased),
it is always the minimum possible letter consistent with the
least possible currently unfilled positions $(i,j)$
that is used --- by virtue of the fact that the entries in $C$
are maintained in increasing order of $\lambda$.
Thus any automorphism of the alphabet $\Sigma$
other than the identity
would yield a larger string.
We conclude that within the main \bwhile loop the assignments
maintain lexicographical order $\prec$
as defined in Section~\ref{sect-prelim}.

It may happen, however, that certain positions $i$ in $\s{x}$ remain empty,
those corresponding to isolated vertices $i$ in $\mathcal{P}^+$.
Assignments to these positions are handled by the final \bfor loop,
which we now consider.
\begin{lemma}
\label{lemm-unassigned}
A vertex $i \in 1..n$ is isolated in $\mathcal{P}^+$ if and only if
\begin{itemize}
\item[(a)]
$\s{y}[i] = 0$ or $i = 1$; \band
\item[(b)]
for every $j \in 2..n$, $\s{y}[j] < i$; \band
\item[(c)]
for every $j \in 1..i\- 1$, $j\+ \s{y}[j] \le i$.
\end{itemize}
\end{lemma}
\begin{proof}
First suppose that $i$ is isolated.
Then (a) must hold; otherwise, for $i > 1$ and $\s{y}[i] > 0$,
there exists an edge $(1,i) \in E^+$, a contradiction.
If (b) does not hold, there exists $j > 1$ such that
$\s{y}[j] = r \ge i$, implying $\s{x}[1..r] \match \s{x}[j..j\+ r\- 1]$,
hence $\s{x}[i] \match \s{x}[j\+ i\- 1]$,
so that $(i,j\+ i\- 1) \in E^+$, again a contradiction.
Similarly, if (c) does not hold, there exists $j \in 1..i\- 1$
such that $\s{y}[j] = r$ with $j\+ r > i$.
Consequently, $\s{x}[j..j\+ r\- 1] \match \s{x}[1..r]$
implying $\s{x}[i] \match \s{x}[i\- j\+ 1]$,
so that $(i\- j\+ 1,i) \in E^+$, a contradiction that establishes sufficiency.

\medskip\noindent
Suppose then that conditions (a)-(c) all hold.
If we assume that $i = 1$ in (a), then (b) implies that
for every $j \in 2..n$, $\s{y}[j] = 0$, so that $E^+ = \emptyset$
and so every position $i$ is isolated.
Otherwise, if $\s{y}[i] = 0$ for some $i > 1$,
conditions (b) and (c) ensure that position $i$ is not contained
in any matching range within \s{x} and is therefore isolated in $\mathcal{P}^+$,
as required. \qed
\end{proof}

Since in the main \bwhile loop
new maximum letters $\lambda_{\max}$ are introduced into pairs
$i$ and $j > i$ of positions in \s{x} that are
determined by entries in \s{y},
it is an immediate consequence of Lemma~\ref{lemm-unassigned}
that $i$
must be less than any isolated vertex,
in particular the smallest one, $i_{\min}$, say.
In other words, every letter assigned during the execution of the main \bwhile loop
occurs at least once to the left of $i_{\min}$ in \s{x}.
It follows that lexicographical order will be maintained
if any required additional letters $\lambda_{\max}\+ 1, \lambda_{\max}\+ 2, \cdots$
are assigned to an ascending sequence of positions in \s{x}
determined by the isolated vertices in $\mathcal{P}^+$.
We have
\begin{lemma}
Given a prefix graph $\mathcal{P}_{\s{y}}$ corresponding to a feasible array $\s{y}$, Algorithm RevEng constructs a lexicographically least indeterminate string
on a minimum alphabet whose prefix table $\s{\pi} = \s{y}$.
\end{lemma}

\subsection{Asymptotic Complexity}
The main \bwhile loop in Algorithm RevEng will be executed exactly $|E^+|$
times.
Within the loop the construction of $C$
requires time proportional to $|C| = |\s{x}[i]|\+ |\s{x}[j]|$,
thus $O(\sigma)$ in the worst case.
The processing of $C$ then also requires $O(\sigma)$ time
in the worst case, except for the time required for {\tt UPDATE\_F}.
Each of the three calls of {\tt UPDATE\_F}
corresponds to the assignment of a letter $\lambda$ to a vertex $i$ of $\mathcal{P}^-$
and the ensuing update of $F[i,\lambda]$,
an event that can occur at most $\sigma$ times for each edge in $E^-$,
thus at most $(\sigma \times |E^-|)$ times overall.
Similarly, the \bfor loop that initializes the $F$ array
requires $\Theta(n)$ time for each of at most $\sigma$ values of $\lambda_{\max}$.

We conclude that the worst-case time requirement for the \bwhile loop
is $O\big(\sigma\max(|E^+|,|E^-|)\big)$.
As illustrated by the examples
$\s{y} = n0^{n-1}$ and $\s{y} = n|n\- 1|\ldots|1$,
the bounds on these quantities are as follows:
$0 \le |E^+|\le \binom{n}{2}$ and $0 \le |E^-| \le n\- 1$,
while $|E^+|\+ |E^-| \ge n\- 1$.
As noted earlier, it was shown in \cite{BBD14} that
$\sigma \le n\+ \sqrt{n}$, with a further conjecture that
$\sigma \le n$.

Turning our attention to the terminating \bfor loop of Algorithm ASSIGN,
we observe that for at most $n$ executions of function {\tt LEAST}, the binary array $B$
of length at most $\sigma$ must be created,
thus overall consuming $O(\sigma n)$ time.
The nested $\forall$ loops in {\tt LEAST} set positions in $B$ $\lambda$ times
for at most every edge in $E^-$, again requiring at most
$O(\sigma n)$ time over all invocations of {\tt LEAST}.
Thus
\begin{lemma}
Algorithm RevEng requires $O(\sigma n^2)$ time in the worst case,
where $\sigma \le n\+ \sqrt{n}$.
\end{lemma}

\subsection{Example}
Suppose $\s{y} = 50210$, so that
$E^+ = 13,14,24$ and $E^- = 12,15,25,35$.
\begin{itemize}
\item
In $E^+$ first consider edge $13$,
leading to assignments $\s{x}[1] \la a,\ \s{x}[3] \la a$,
with $F[2,a] = F[5,a] = 1$ since both $12$ and $15$ are edges of $E^-$.
\item
Edge $14$ of $E^+$ leads to $\s{x}[4] \la a$ and no new values in $F$,
since vertex $4$ is isolated in $E^-$.
\item
Edge $24$ of $E^+$ requires a new letter because $F[2,a] = 1$.
Therefore we assign $\s{x}[2] \la b$ and $\s{x}[4] \stackrel{+}{\la} b$,
while setting $F[1,b] = F[5,b] = 1$ because of the edges
$21$ and $25$ in $E^-$.
\item
Finally we deal with the isolated vertex $5$ in $E^+$
by setting $\s{x}[5] \la c$ since $15 \in E^-$ and $\s{x}[1] = a$,
while $25 \in E^-$ and $\s{x}[2] = b$.
\end{itemize}
The lexicographically least string
is $\s{x} = a b a \{a,b\} c$.

\subsection{Computational Experiments}
To get an idea of how the algorithm behaves in practice, we have implemented
Algorithm RevEng and conducted a simple experimental study. A set of 1000
feasible arrays having lengths $10, 20, .., 100$ has been randomly generated as
follows. For each feasible array $\s{y}$ we randomly select a value for
$\s{y}[i], i \in [1..n]$ from within the range $[0..n-i+1]$. The experiments
have been run on a Windows Server 2008 R2 64-bit Operating System, with
Intel(R) Core(TM) i7 2600 processor @ 3.40GHz having an installed memory (RAM)
of 8.00 GB. We have implemented Algorithm RevEng in $C\#$ language using Visual
Studio 2010. As is evident from Figure~\ref{fig-tests}, the experiments suggest
that average case time also increases by a factor somewhat greater than $n^2$.
\begin{figure}[t!]
  \centering
  \includegraphics*[scale=1]{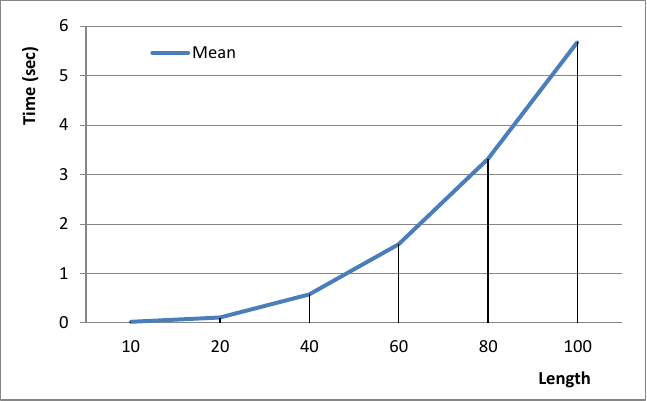}
  \caption{Timing results for randomly-generated feasible arrays
  \s{y}.}
  \label{fig-tests}
\end{figure}
\section{Discussion}
\label{sect-future}
The high worst-case time complexity of the algorithm described here
suggests room for improvement. 
On the other hand, it is difficult to imagine an algorithm
that could do the same computation without considering all the edges
of $E^+$ and thus necessitating $\Theta(n^2)$ time for many instances of
the prefix table $\s{\pi}$.
Similarly, the requirement to achieve a lexicographically least solution
leads to a recurring dependence on alphabet size $\sigma$
that expresses itself in the time complexity.
Even though it may be true that $\sigma \le n$,
nevertheless it seems clear that $\sigma$ can be much larger
than in the regular case, where it has been shown \cite{CCR09,CRSW13}
that $\sigma \le \ceil{\log_2 n}$.

We have tried approaches that focus on $E^-$ rather than $E^+$
as the primary data structure, but without success.
In particular, we have considered ``triangles'' $i_1ji_2$,
where both $(i_1,j)$ and $(i_2,j)$ are edges in $E^+$,
while $(i_1,i_2) \in E^-$, a situation that forces a
string to be indeterminate.
It turns out, however, that the number of such triangles is $O(n^2)$.
Similarly, the ingenious graph proposed in \cite{BBD14}, 
whose chromatic number is the minimum alphabet size $\sigma$
of a string corresponding to a given prefix table,
has $O(n^2)$ vertices in the worst case. 

At the same time, we have no proof that our algorithm
is asymptotically optimal; for example, an algorithm that could eliminate
the $\sigma$ factor in the complexity would be of considerable interest.
Also interesting would be an algorithm for indeterminate strings that would
execute in $\Theta(n)$ time on regular strings as a special case, thus matching the algorithm of \cite{CCR09}.
More generally, we propose the study of indeterminate strings
(``strings'' as we have called them here),
their associated data structures (such as the prefix table),
and their applications as a promising research area
in both combinatorics on words and string algorithms.

\bibliographystyle{alpha} 
\bibliography{references.bib}
 
\end{document}